\documentclass[english]{article}
\usepackage{amsfonts,amsmath,amsthm,amscd,amssymb,latexsym}
\newcommand{\ep}{\varepsilon}

\newcommand{\vph}{\varphi}

\newcommand{\intl}{\int\limits}

\newtheorem{thm}{Theorem}
\newtheorem{cor}[thm]{Corollary}

\newtheorem{lem}[thm]{Lemma}
\newtheorem{prop}[thm]{Proposition}
\newtheorem{defn}[thm]{Definition}
\numberwithin{equation}{section}
\newcommand{\keywords}{\textbf{Key words and phrases: }\medskip}
\newcommand{\subjclass}{\textbf{Math. Subj. Clas.: }\medskip}

\begin{document}
\title{\textbf{ Green function for a two-dimensional discrete Laplace-Beltrami
operator}}
\author{\textbf{Volodymyr Sushch}}
\date{Koszalin University of Technology \\
  Sniadeckich 2, 75-453 Koszalin, Poland \\
  volodymyr.sushch@tu.koszalin.pl}  \maketitle

\begin{abstract}
 We study a discrete model of the Laplacian  in $\mathbb{R}^2$
 that preserves the geometric structure of
the original continual object. This means that,
speaking of a discrete model, we do not mean just  the direct
replacement of differential operators by  difference ones  but
also a discrete analog of the Riemannian structure. We consider this structure on
the appropriate combinatorial analog of differential forms. Self-adjointness and boundness for a discrete Laplacian
 are proved. We define the Green function for this operator  and also derive  an explicit formula of the one.
\end{abstract}
\section{Introduction}
We begin with a brief review of some well known definitions   that are related to the contents of this paper.
Denote by $\Lambda^r(\mathbb{R}^2)$ the set of all differentiable complex-valued
$r$-forms on $\mathbb{R}^2$, where  $r=0,1,2$. Let $\ast: \Lambda^r(\mathbb{R}^2)\rightarrow\Lambda^{2-r}(\mathbb{R}^2)$
 be the Hodge star operator.  An inner product for $r$-forms  with compact support is defined by
\begin{equation} \label{1.1}(\vph, \
\psi)=\intl_{\mathbb{R}^2}\vph\wedge\ast\overline{\psi},
\end{equation}
where the bar over $\psi$ denotes complex conjugation.  Let $L^2\Lambda^r(\mathbb{R}^2)$ denote the completion of
$\Lambda^r(\mathbb{R}^2)$ with respect to the norm generated by the inner product (\ref{1.1}). Let the exterior derivative
$d:\Lambda^r(\mathbb{R}^2)\rightarrow
\Lambda^{r+1}(\mathbb{R}^2)$ be defined as usual.
We define the operator
 \begin{equation}\label{1.2}
 d:L^2\Lambda^r(\mathbb{R}^2)\rightarrow
L^2\Lambda^{r+1}(\mathbb{R}^2)
 \end{equation}
as  the closure in  the $L^2$-norm the corresponding operation specified on smooth forms.
The adjoint of $d$, denoted by $\delta$, is given by
 \begin{equation*}(d\vph, \ \omega)=(\vph, \
\delta\omega), \qquad \vph\in L^2\Lambda^r(\mathbb{R}^2), \quad \omega\in
L^2\Lambda^{r+1}(\mathbb{R}^2).
\end{equation*}
Note  that $\delta:L^2\Lambda^{r+1}(\mathbb{R}^2)\rightarrow L^2\Lambda^r(\mathbb{R}^2)$. The following relations hold among
$\ast$, $d$ and $\delta$. See for instance \cite{W}.
\begin{equation}\label{1.3}
 \ast^2=(-1)^{r(2-r)}Id, \qquad \delta=(-1)^r\ast^{-1}d\ast.
\end{equation}
 The Laplacian is defined to be
 \begin{equation}\label{1.4}
 -\Delta=d\delta+\delta d:L^2\Lambda^r(\mathbb{R}^2)\rightarrow
L^2\Lambda^{r}(\mathbb{R}^2).
\end{equation}

 In this paper we develop some combinatorial structures that are analogs of objects in differential geometry.
We are interested in finding of a natural discrete analog of the Laplacian on cochains.
Speaking of a discrete model, we  mean not only  the direct
replacement of differential operators by  difference ones  but
also  discrete analogs of all essential ingredients of the Riemannian structure over a
properly introduced  combinatorial object.

Our approach bases on the formalism proposed by Dezin \cite{Dezin}. We adapt the
combinatorial constructions from \cite{Dezin, S1, S2} and define discrete analogs of operators (\ref{1.2})--(\ref{1.4})
in a similar way. In \cite{Dezin}, Dezin study  discrete Laplace operators in finite-dimensional Hilbert spaces, i.e. on cochains given in domains  with boundary.
In this paper we extend these results on an infinite complex of complex-valued cochains. We prove self-adjointness and
boundness  for the discrete
Laplace-Beltrami operator in infinite Hilbert spaces that are associated  with $L^2\Lambda^r(\mathbb{R}^2)$. Spectral properties are discussed.
We define the Green function for the discrete Laplacian and derive the one in an explicit form.

There are  other  geometric approaches to discretisation of the Hodge theory of harmonic forms
presented in \cite{D, DP, Kom}.
In all these papers discrete models are given on the simplicial cochains of triangulated  closed Riemannian manifolds.
See also \cite{Bob, Chung, M} and references given there.

Classical references on second order difference equations are the books by  Berezanski \cite[ch. 7]{B}, Atkinson \cite{A}
and the most recent monograph by Teschl \cite{T}.

\section{Combinatorial structures}

Let us denote by $\mathfrak{C}(2)$ the two-dimensional complex. This
complex is defined by
$$\mathfrak{C}(2)=\mathfrak{C}^0\oplus\mathfrak{C}^1\oplus\mathfrak{C}^2,$$
where $\mathfrak{C}^r$ is a real linear space of $r$-dimensional
chains. We follow the  notation of \cite{Dezin, S2}. Let $\{x_{k,s}\}$, \ $\{e_{k,s}^1, \ e_{k,s}^2\},$ \
$\{\Omega_{k,s}\},$ \ $k, s\in\mathbb{Z},$  be the sets of basis
elements of $\mathfrak{C}^0, \ \mathfrak{C}^1, \ \mathfrak{C}^2$
respectively. It is convenient to introduce shift operators $$\tau
k=k+1, \qquad \sigma k=k-1$$ in the set of indices. The
boundary operator $\partial$ is defined by the rule
\begin{align}\label{2.1}\notag
\partial x_{k,s}=0, \quad \partial \Omega_{k,s}=e_{k,s}^1+e_{\tau k,s}^2-e_{k,\tau
s}^1-e_{k,s}^2, \\ \partial e_{k,s}^1=x_{\tau k,s}-x_{k,s}, \quad
\partial e_{k,s}^2=x_{k,\tau s}-x_{k,s}.
\end{align}
The definition of $\partial$ is linearly extended to arbitrary
chains. We call the complex $\mathfrak{C}(2)$ a combinatorial model
of $\mathbb{R}^2$.

On the other hand, we can consider $\mathfrak{C}(2)$ as the tensor
product $\mathfrak{C}(2)=\mathfrak{C}\otimes \mathfrak{C}$ of the
one-dimensional complex $\mathfrak{C}$ (combinatorial model of a
real line). Then basis elements of $\mathfrak{C}(2)$ can be
written as follows
\begin{align*}
x_k\otimes x_s=x_{k,s}, &\qquad e_k\otimes x_s=e_{k,s}^1,\\
e_k\otimes e_s=\Omega_{k,s}, &\qquad x_k\otimes
e_s=e_{k,s}^2,\end{align*} where $x_k$, $e_k$ are the basis elements
of $\mathfrak{C}$.

 Let us introduce an object dual to $\mathfrak{C}(2)$. Name\-ly,
the complex of complex-valued functions over $\mathfrak{C}(2)$. The
dual complex $K(2)$ we can consider as the set of complex-valued
cochains and it has the same structure as $\mathfrak{C}(2)$, i.e.
$K(2)=K^0\oplus K^1\oplus K^2$. In other words, $K(2)$ is a linear
complex space with basis elements
$$\{x^{k,s}, \ e^{k,s}_1, \ e^{k,s}_2, \ \Omega^{k,s}\}.$$

The  pairing (chain-cochain) operation is defined by the rule:
\begin{equation}\label{2.2}
<x_{k,s}, \ x^{p,q}>=<\Omega_{k,s}, \ \Omega^{p,q}>=<e_{k,s}^1, \
e^{p,q}_1>=<e_{k,s}^2, \ e^{p,q}_2>=\delta_{k,s}^{p,q},
\end{equation}
where $\delta_{k,s}^{p,q}$ is  Kronecker symbol. We call elements of
the complex $K(2)$ forms, emphasizing their closeness to the
corresponding continual objects, differential forms.
 Then the 0-, 1-, 2-forms $\vph, \
\omega=(u, v), \ \eta$ can be written as
\begin{equation}\label{2.3} \vph=\sum_{k,s}\vph_{k,s}x^{k,s},
\qquad \eta=\sum_{k,s}\eta_{k,s}\Omega^{k,s}, \qquad
\omega=\sum_{k,s}(u_{k,s}e^{k,s}_1+v_{k,s}e^{k,s}_2),
\end{equation}
 where $\vph_{k,s}, \ u_{k,s}, \
v_{k,s},  \ \eta_{k,s}\in\mathbb{C}$ for any $k,s\in\mathbb{Z}$.
Operation (\ref{2.2}) is extended to arbitrary forms (\ref{2.3}) by
linearity. The boundary operator (\ref{2.1}) in $\mathfrak{C}(2)$
induces the dual operation $d^c$ in $K(2)$:
\begin{equation}\label{2.4}
<\partial a, \ \alpha>=<a, \ d^c\alpha>,
\end{equation}
where $a\in \mathfrak{C}(2), \ \alpha\in K(2)$. We assume that the
 coboundary operator $d^c: K^r \rightarrow K^{r+1}$ is a
discrete analog of the exterior differentiation operator $d$ (\ref{1.3}).

If $\vph\in K^0$ and $\omega=(u, v)\in K^1$, then we have the
following difference representations for $d^c$:
\begin{align}\label{2.5}
 &<e^1_{k,s}, \ d^c\vph>=\vph_{\tau k,
 s}-\vph_{k,s}\equiv\Delta_k\vph_{k,s},\notag \\
 & <e^2_{k,s}, \ d^c\vph>=\vph_{k, \tau
 s}-\vph_{k,s}\equiv\Delta_s\vph_{k,s}\notag \\
  &<\Omega_{k,s}, \ d^c\omega>=v_{\tau k,
 s}-v_{k,s}-u_{k,\tau s}+u_{k,s}\equiv\Delta_k v_{k,s}-\Delta_s
 u_{k,s}.
 \end{align}
 Note that if $\eta\in K^2$, then $d^c\eta=0$.

 Let us now introduce in $K(2)$ a multiplication which is an analog of the exterior
 multiplication $\wedge$ for differential forms. We denote this operation by
 $\cup$ and define it according to the rule:
\begin{align}\label{2.6}
&x^{k,s}\cup x^{k,s}=x^{k,s}, \qquad e^{k,s}_2\cup e^{k,\tau
s}_1=-\Omega^{k,s},\notag\\  &x^{k,s}\cup e^{k,s}_1=e^{k,s}_1\cup
x^{\tau k,s}=e^{k,s}_1,\notag \\ &x^{k,s}\cup
e^{k,s}_2=e^{k,s}_2\cup x^{k,\tau s}=e^{k,s}_2,\notag
\\ &x^{k,s}\cup \Omega^{k,s}= \Omega^{k,s}\cup x^{\tau k,\tau s}=e^{k,s}_1\cup
e^{\tau k,s}_2=\Omega^{k,s},
 \end{align}
 supposing the product to be zero in all other cases.
 The $\cup$-multiplication is extended to discrete
 forms by linearity.
 In terms of the theory of homologies, this is the so-called Whitney multiplication.
 For arbitrary forms $\alpha, \beta\in K(2)$ we have
  the following relation
 \begin{equation}\label{2.7}
d^c(\alpha\cup\beta)=d^c\alpha\cup\beta+(-1)^r\alpha\cup d^c\beta,
\end{equation}
where $r$ is the degree of  $\alpha$. The proof of this
can be found in Dezin \cite[p.~147]{Dezin}. Relation (\ref{2.7}) is
an analog of the corresponding continual relation for differential
forms (see \cite{W}).

 Define a discrete analog of the  Hodge star operator.
 Let $\ep^{k,s}$ denote an arbitrary basis element of $K(2)$. We
introduce the   operation $\ast: K^r\rightarrow K^{2-r}$ by setting
\begin{equation}\label{2.8}
\ep^{k,s}\cup\ast\ep^{k,s}=\Omega^{k,s}.
\end{equation}
Using (\ref{2.6}) we get
\begin{equation*}
\ast x^{k,s}=\Omega^{k,s},  \quad \ast e^{k,s}_1=e^{\tau k,s}_2,
\quad \ast e^{k,s}_2=-e^{k,\tau s}_1,  \quad \ast
\Omega^{k,s}=x^{\tau k,\tau s}.
\end{equation*}
The operation $\ast$ is extended to arbitrary forms by linearity.

Let $\alpha\in K^r$ is an arbitrary  $r$-form:
\begin{equation}\label{2.9}
\alpha=\sum_{k,s}\alpha_{k,s}\ep^{k,s}.
\end{equation}
Denote by $K^r_0$ the set of all discrete
 $r$-form with compact support on $\mathfrak{C}(2)$.
 Let $\Omega$  be the following "domain"
 \begin{equation}\label{2.10}
\Omega=\sum_{k,s}\Omega_{k,s},  \qquad k,s\in\mathbb{Z},
\end{equation}
where $\Omega_{k,s}$ is a two-dimensional basis element of
 $\mathfrak{C}(2)$.
 Note that if the sum  (\ref{2.10})
 is finite and let \ $-N\leq k,s\leq N,$  \  $N\in\mathbb{N}$, then we will  write  $\Omega=\Omega_N$.

 The relation
\begin{equation}\label{2.11}
(\alpha, \ \beta)=<\Omega, \ \alpha\cup\ast\overline{\beta}>,
\end{equation}
where $\alpha, \beta\in K^r_0$, gives a correct definition of
inner product in $K(2)$. Using (\ref{2.2}), (\ref{2.6}) and
(\ref{2.8}), this definition can be rewritten as follows
\begin{equation}\label{2.12}
(\alpha, \ \beta)=\sum_{k,s}\alpha_{k,s}\overline{\beta_{k,s}}.
\end{equation}
For  $\Omega=\Omega_N$ we will write
\begin{equation*}
(\alpha, \ \beta)_N=<\Omega_N, \
\alpha\cup\ast\overline{\beta}>=\sum_{k,s=-N}^N\alpha_{k,s}\overline{\beta_{k,s}}.
\end{equation*}
Let $\alpha\in K^r, \ \beta\in K^{r+1}$. The relation
\begin{equation}\label{2.13}
(d^c\alpha, \ \beta)_N=<\partial\Omega_N, \
\alpha\cup\ast\overline{\beta}>+(\alpha, \ \delta^c\beta)_N
\end{equation}
defines the operator $\delta^c: K^{r+1}\rightarrow K^r,$ \
\begin{equation*}\delta^c\beta=(-1)^r\ast^{-1}d^c\ast\beta,
\end{equation*}
which is the formally adjoint operator to $d^c$ (see \cite {Dezin} for more details). It
is obvious that the operator $\delta^c$ can be regarded as a
discrete analog of the codifferential $\delta$ (cf.  (\ref{1.3})).  Equation (\ref{2.13})
is an analog of the Green formula for the formally adjoint differential
operators $d$ and $\delta$.
 It is easy to check that for $\alpha\in
K^r_0, \quad \beta\in K^{r+1}_0$ we obtain
\begin{equation}\label{2.14}
(d^c\alpha, \ \beta)=(\alpha, \ \delta^c\beta).
\end{equation}
According to (\ref{2.5}),  we have  $\delta^c\varphi=0$ and
\begin{align}\label{2.15, 2.16}
&<x_{k,s}, \ \delta^c\omega>=-\Delta_k u_{\sigma k,s}-\Delta_s v_{k,\sigma s},\\
&<e^1_{k,s}, \ \delta^c\eta>=\Delta_s\eta_{k,\sigma s},\qquad
<e^2_{k,s}, \ \delta^c\eta>=-\Delta_k\eta_{\sigma k, s},
\end{align}
where $\varphi\in K^0, \ \omega\in K^1$ and $\eta\in K^2$.

Therefore a discrete analog of the Laplace-Beltrami  operator (\ref{1.4}) can be defined
as follows
\begin{equation}\label{2.17}
-\Delta^c=\delta^c d^c+d^c\delta^c: K^r\rightarrow K^r.
\end{equation}
Obviously, if $\vph\in K^0$, then we have
\begin{equation}\label{2.18}
-\Delta^c\vph=\delta^c d^c\vph.
\end{equation}
Combining (\ref{2.15, 2.16}) with (\ref{2.5}) we can rewrite (\ref{2.18}) as
\begin{equation}\label{2.19}
<x_{k,s}, \ -\Delta^c\vph>=4\vph_{k,s}-\vph_{\tau
k,s}-\vph_{k,\tau s}-\vph_{\sigma k,s}-\vph_{k,\sigma s}.
\end{equation}
The same difference form of (\ref{2.17}) can be drawn for the components $\eta_{k,s}$ of $\eta\in K^2$ and for the two
components $u_{k,s}  \ v_{k,s}$ of  $\omega\in K^1$.

\section{Discrete Laplacian}

Let us now introduce the linear space
\begin{equation}\label{3.1}
\mathcal{H}^r=\{\alpha\in K^r: \
\sum_{k,s}|\alpha_{k,s}|^2<+\infty, \quad k,s\in\mathbb{Z}\},
 \end{equation}
 where  $r=0,1,2$.
 Clearly,   $\mathcal{H}^r$ is a Hilbert space with  inner
 product (\ref{2.11}) (or (\ref{2.12})) and with the following norm
 \begin{equation}\label{3.2}
\|\alpha\|=\sqrt{(\alpha, \
\alpha)}=\Big(\sum_{k,s}|\alpha_{k,s}|^2\Big)^{\frac{1}{2}}.
\end{equation}
Note that if  $\alpha\in\mathcal{H}^r$, then the set of complex-valued sequences $(\alpha_{k,s})$ is
  $\ell^2(\mathbb{Z}^2)$.
From now on we regard $d^c, \ \delta^c$ and  $-\Delta^c$ as the following operators
\begin{equation*}
 d^c: \mathcal{H}^r\rightarrow \mathcal{H}^{r+1}, \qquad \delta^c: \mathcal{H}^r\rightarrow \mathcal{H}^{r-1},
 \qquad -\Delta^c: \mathcal{H}^r\rightarrow \mathcal{H}^r,
\end{equation*}
where $r=0,1,2$. It is convenient to suppose that  $\mathcal{H}^{-1}=\mathcal{H}^3=0$.

\begin{thm}
The  operators
\begin{equation}\label{3.3}
 -\Delta^c: \mathcal{H}^r\rightarrow \mathcal{H}^r, \qquad r=0,1,2,
 \end{equation}
  are bounded and self-adjoint. Moreover, $\|-\Delta^c\|=8,$
where $\|-\Delta^c\|$ denotes the operator norm of $-\Delta^c$.
\end{thm}
\begin{proof}
We begin by proving self-adjointness of $-\Delta^c$ for the case $r=0$. Let $\vph, \psi\in\mathcal{H}^0$ and
$\omega=(u,v)\in\mathcal{H}^1$. Then the Green formula (\ref{2.13}) can be rewritten as
\begin{align}\label{3.4}\notag
(d^c\vph, \ \omega)_N&=\sum_{k=-N}^N(\vph_{k,\tau N}\overline{v_{k,N}}-\vph_{k,-N}\overline{v_{k,-\tau N}})+\\
&+\sum_{s=-N}^N(\vph_{\tau N,s}\overline{u_{N,s}}-\vph_{-N,s}\overline{u_{-\tau N,s}})+(\vph, \ \delta^c\omega)_N
\end{align}
The substitution of $d^c\psi$ for $\omega$  in (\ref{3.4}) gives
\begin{align}\label{3.5}\notag
(d^c\vph, \ d^c\psi)_N=\sum_{k=-N}^N\big(\vph_{k,\tau N}(\overline{\psi_{k,\tau N}}-\overline{\psi_{k,N}})-
\vph_{k,-N}(\overline{\psi_{k,-N}}-\overline{\psi_{k,-\tau N}})\big)+\\
+\sum_{s=-N}^N\big(\vph_{\tau N,s}(\overline{\psi_{\tau N,s}}-\overline{\psi_{N,s}})
-\vph_{-N,s}(\overline{\psi_{-N,s}}-\overline{\psi_{-\tau N,s}})\big)+(\vph, \ \delta^c d^c\psi)_N.
\end{align}
Letting $N\rightarrow+\infty$ we get
\begin{equation}\label{3.6}
(d^c\vph, \ d^c\psi)=(\vph, \ \delta^c d^c\psi).
\end{equation}
It follows immediately that $-\Delta^c: \mathcal{H}^0\rightarrow \mathcal{H}^0$ is self-adjoint.

The same proof remains valid for the case $r=2$. Now we have $-\Delta^c=d^c\delta^c$
and the analog of relation (\ref{3.5}) results from the inner product $(\delta^c\eta, \ \delta^c\zeta)_N$,
where $\eta, \zeta\in\mathcal{H}^2$. A trivial verification shows that properties of $(\delta^c\eta, \ \delta^c\zeta)_N$
are completely similar to those of $(d^c\vph, \ d^c\psi)_N$. Hence
\begin{equation}\label{3.7}
(\delta^c\eta, \ \delta^c\zeta)=(\eta, \ d^c\delta^c\zeta).
\end{equation}

Finally,  let $r=1$. In this case we have $-\Delta^c=d^c\delta^c+\delta^c d^c$  and we must study the sum
$(d^c\omega, \ d^c\vartheta)+(\delta^c\omega, \ \delta^c\vartheta)$, where $\omega=(u,v)\in\mathcal{H}^1$, \
$\vartheta=(f,g)\in\mathcal{H}^1$. Taking in (\ref{2.13}) $\alpha=\omega$ and $\beta=d^c\vartheta$ we obtain the analog of
relation (\ref{3.5}) for 1-forms
\begin{align*}
&(d^c\omega, \ d^c\vartheta)_N\equiv\sum_{k,s=-N}^N(\Delta_k v_{k,s}-\Delta_s u_{k,s})(\overline{\Delta_k g_{k,s}-\Delta_s f_{k,s}})=\\
&=\sum_{k=-N}^N\big[u_{k,-N}(\overline{\Delta_kg_{k,-\tau N}}-\overline{\Delta_Nf_{k,-\tau N}})-
u_{k,\tau N}(\overline{\Delta_kg_{k,N}}-\overline{\Delta_Nf_{k,N}})\big]+\\
&+\sum_{s=-N}^N\big[v_{\tau N,s}(\overline{\Delta_Ng_{N,s}}-\overline{\Delta_sf_{N,s}})-
v_{-N,s}(\overline{\Delta_Ng_{-\tau N,s}}-\overline{\Delta_sf_{-\tau N,s}})\big]+\\
&+(\omega, \ \delta^c d^c\vartheta)_N.
\end{align*}
Letting $N\rightarrow+\infty$ we obtain equation (\ref{3.6}) for the 1-forms $\omega, \ \vartheta\in\mathcal{H}^1$.
In the same manner we can see that equation (\ref{3.7}) holds for $\omega, \ \vartheta\in\mathcal{H}^1$.
Adding we obtain
\begin{equation}\label{3.8}
(d^c\omega, \ d^c\vartheta)+(\delta^c\omega, \ \delta^c\vartheta)=(\omega, \ -\Delta^c\vartheta).
\end{equation}
Thus it follows that
\begin{equation*}
(-\Delta^c\omega, \ \vartheta)=(\omega, \ -\Delta^c\vartheta).
\end{equation*}

For the rest of the proof let $\alpha\in\mathcal{H}^r$ be an arbitrary $r$-form. Substituting  (\ref{2.19}) into (\ref{2.12}) we get
\begin{align*}
|(-\Delta^c\alpha, \ \alpha)|&=\Big|\sum_{k,s}(4\alpha_{k,s}-\alpha_{\tau
k,s}-\alpha_{k,\tau s}-\alpha_{\sigma k,s}-\alpha_{k,\sigma s})\overline{\alpha_{k,s}}\Big|\leq\\
&\leq 4\sum_{k,s}|\alpha_{k,s}|^2+\sum_{k,s}|\alpha_{\tau
k,s}\overline{\alpha_{k,s}}|+\sum_{k,s}|\alpha_{k,\tau s}\overline{\alpha_{k,s}}|+\\
&+\sum_{k,s}|\alpha_{\sigma k,s}\overline{\alpha_{k,s}}|+\sum_{k,s}|\alpha_{k,\sigma s}\overline{\alpha_{k,s}}|
\leq 8\|\alpha\|^2.
\end{align*}
From this we conclude that $\|-\Delta^c\|\leq8$. Since $-\Delta^c$ is self-adjoint,
it follows easily that $\|-\Delta^c\|=8$ (see for instance \cite[Ch. 3]{Mi}).

\end{proof}

\begin{cor} The  operators  (\ref{3.3}) are positive, i.e. for any non-trivial
$r$-form $\alpha\in\mathcal{H}^r$ we have
 \begin{equation*}
 (-\Delta^c\alpha, \ \alpha)>0.
 \end{equation*}
 \end{cor}
\begin{proof}
This follows from (\ref{3.6}), (\ref{3.7}) and (\ref{3.8}).
\end{proof}

\begin{cor} For any $r$, $r=0,1,2$, we have
 \begin{equation*}
 \sigma(-\Delta^c)=[0,\ 8],
 \end{equation*}
 where $\sigma(-\Delta^c)$ denotes the spectrum of $-\Delta^c$.
 \end{cor}
 \begin{proof}
Straightforward.
\end{proof}

\section{Discrete analog of the Green function}

Let $\varrho(-\Delta^c)=\mathbb{C}\setminus\sigma(-\Delta^c)$ denotes the resolvent set of  $-\Delta^c$.
In this section we try to describe the resolvent  operator $(-\Delta^c-\lambda)^{-1}$, \
$\lambda\in\varrho(-\Delta^c)$, \ of the operator $-\Delta^c: \mathcal{H}^0\rightarrow \mathcal{H}^0$.
Let us introduce a discrete form $G(x,\tilde{x})$ on $\mathfrak{C}^0\times\mathfrak{C}^0$  as follows
\begin{equation*}
 G(x,\tilde{x})=\sum_{k,s}G_{k,s}(\tilde{x})x^{k,s}, \quad  \mbox{where} \quad
  G_{k,s}(\tilde{x})=\sum_{m,n}G_{k,s,m,n}x^{m,n}
 \end{equation*}
 and $G_{k,s,m,n}\in\mathbb{C}$  for any $k,s,m,n\in\mathbb{Z}$.  Hence we have
\begin{equation}\label{4.1}
 G(x,\tilde{x})=\sum_{k,s}\sum_{m,n}G_{k,s,m,n}x^{k,s}x^{m,n}.
 \end{equation}
 This is a so-called discrete double form (for details see \cite{S1}). Note that in the continual case
(for differential forms) this construction is
 due to  De~Rham \cite{Rham}.  It is obvious that the basis elements $x^{k,s}$ and $x^{m,n}$
 in (\ref{4.1}) commute and $G(x,\tilde{x})=G(\tilde{x},x)$.
 By analogy with (\ref{2.2}), the double 0-form $G(x,\tilde{x})$ can be written pointwise as
\begin{equation}\label{4.2}
 <x_{k,s}x_{m,n}, \ G(x,\tilde{x})>=G_{k,s,m,n}.
 \end{equation}
Let $\varphi\in\mathcal{H}^0$. Then we define a 0-form $\delta^{m,n}$, $m,n\in\mathbb{Z},$ by setting
\begin{equation}\label{4.3}
 (\varphi, \ \delta^{m,n})=\sum_{k,s}\varphi_{k,s}\delta^{m,n}_{k,s}=\varphi_{m,n},
 \end{equation}
 where $\delta^{m,n}_{k,s}$ is the Kronecker delta. By analogy with the continual case, the 0-form $\delta^{m,n}$
 defined by (\ref{4.3}) will be called a discrete analog of  Dirac's $\delta$-function at the point $x_{m,n}$.
 We can write $\delta^{m,n}$ as
 \begin{equation*}\delta^{m,n}=\sum_{k,s}\delta^{m,n}_{k,s}x^{k,s}=x^{m,n}.
 \end{equation*}
 We need also the following double form
\begin{equation}\label{4.4}
 \delta(x,\tilde{x})=\sum_{k,s}\sum_{m,n}\delta^{m,n}_{k,s}x^{k,s}x^{m,n}=
 \sum_{m,n}\delta^{m,n}x^{m,n}=\sum_{m,n}x^{m,n}x^{m,n}.
 \end{equation}
\begin{defn}
The double form (\ref{4.1}) is called the  Green function  for the operator
$-\Delta^c: \mathcal{H}^0\rightarrow \mathcal{H}^0$ if
\begin{equation}\label{4.5}
G_{k,s,m,n}(\lambda)=(\delta^{k,s}, \ (-\Delta^c-\lambda)^{-1}\delta^{m,n})
\end{equation}
for any  $k,s,m,n\in\mathbb{Z}$.
\end{defn}
Of course,
\begin{equation*}
(-\Delta^c-\lambda)^{-1}\delta^{m,n}=G_{m,n}(x,\lambda).
\end{equation*}
It follows easily that
\begin{equation}\label{4.6}
(-\Delta^c-\lambda)_xG(x,\tilde{x},\lambda)=\delta(x,\tilde{x}),
\end{equation}
where $(-\Delta^c-\lambda)_x$ is the operator $-\Delta^c-\lambda$ that acts with respect to $x$.
Recall that in our abbreviation  $x$ corresponds to  $k,s$.

We have
\begin{equation*}
(-\Delta^c-\lambda)^{-1}\varphi=\sum_{k,s}\sum_{m,n}G_{k,s,m,n}(\lambda)\varphi_{m,n}x^{k,s}x^{m,n},
\qquad \varphi\in\mathcal{H}^0, \quad \lambda\in\varrho(-\Delta^c).
\end{equation*}
Indeed, applying $(-\Delta^c-\lambda)_x$ to the right-hand side gives
\begin{align*}
(-\Delta^c-\lambda)_x\sum_{m,n}\varphi_{m,n}G_{m,n}(x,\lambda)x^{m,n}&=
\sum_{m,n}\varphi_{m,n}(-\Delta^c-\lambda)_xG_{m,n}(x,\lambda)x^{m,n}=\\
=\sum_{m,n}\varphi_{m,n}\delta^{m,n}x^{m,n}&=\sum_{m,n}\varphi_{m,n}x^{m,n}x^{m,n}=\varphi.
\end{align*}

We now try to write the Green function for $-\Delta^c$ in a somewhat more explicit way.
 For this we construct a solution of the equation
 \begin{equation}\label{4.7}
-\Delta^c\varphi=\lambda\varphi,  \qquad \lambda\in\mathbb{C}.
\end{equation}
The following construction is adapted from \cite{T}, where the Green function is studied  for Jacobi operators.
By (\ref{2.19}), equation (\ref{4.7}) can be written pointwise (at the point $x_{k,s}$) as
\begin{equation}\label{4.8}
4\vph_{k,s}-\vph_{\tau
k,s}-\vph_{k,\tau s}-\vph_{\sigma k,s}-\vph_{k,\sigma s}=\lambda\vph_{k,s}.
\end{equation}
Applying  the transformation $\lambda=-4\mu+4$ we reduce (\ref{4.8}) to the equation
\begin{equation}\label{4.9}
\frac{1}{4}(\vph_{\tau
k,s}+\vph_{k,\tau s}+\vph_{\sigma k,s}+\vph_{k,\sigma s})=\mu\vph_{k,s}.
\end{equation}
An easy  computation shows that, substituting the  ansatz $\varphi_{k,s}=p^{k+s}$ into  (\ref{4.9}), we obtain
\begin{equation}\label{4.10}
\varphi_{k,s}^{\pm}(\mu)=(\mu\pm R(\mu))^{k+s},
\end{equation}
where  $R(\mu)=-\sqrt{\mu^2-1}$ and $\sqrt{\cdot}$ denotes the standard brunch of the square root.
 It follows that   we can write the solutions $\varphi^{\pm}(\lambda)$ of (\ref{4.7})
in the form
\begin{equation}\label{4.11}
\varphi^{\pm}(\lambda)=\sum_{k,s}\varphi_{k,s}^{\pm}(\mu)x^{k,s},
\end{equation}
where $\varphi_{k,s}^{\pm}(\mu)$ are given by (\ref{4.10}) and
$\mu=1-\frac{\lambda}{4}$. Obviously,  $\lambda\in[0,\ 8]$ if and
only if  $\mu\in[-1, 1]$.  By Corollary~3,
$\lambda\in\varrho(-\Delta^c)$ leads to
$\mu\in\mathbb{C}\setminus[-1, 1]$.

It is convenient to write $\varphi_{k,s}^{\pm}(\mu)$ as
\begin{equation}\label{4.12}
\varphi_{k,s}^{\pm}(\mu)=\varphi_k^{\pm}(\mu)\cdot\varphi_s^{\pm}(\mu),
\end{equation}
where $\varphi_{k}^{\pm}(\mu)=(\mu\pm R(\mu))^{k}$  and $k,
s\in\mathbb{Z}$.
  Suppose $\mu\in\mathbb{C}\setminus[-1, 1]$. Then
one has to examine that  the sequences
$(\varphi_k^{\pm}(\mu))_{k\in\mathbb{Z}}$ are square summable  near
$\pm\infty$ respectively, i.e these are $\ell^2(\mathbb{Z})$ near
$\pm\infty$.

Let $\mathcal{H}^0_{\pm}$ denotes the set of 0-forms whose
restriction to $K_{\pm}^0$ belongs to $\mathcal{H}^0$. Here
$K_{+}^0$  ($K_{-}^0$) denotes the set of 0-cochains (\ref{2.9})
with $k+s>0$  ($k+s<0$).
\begin{prop} Let $\lambda\in\varrho(-\Delta^c)$. Then
$\varphi^{\pm}(\lambda)\in\mathcal{H}^0_{\pm}$.
\end{prop}
\begin{proof}
This follows immediately from (\ref{4.12}).
\end{proof}
\begin{lem} For any $k,
s\in\mathbb{Z}$  the component
$\varphi_k^{+}(\mu)\cdot\varphi_s^{-}(\mu)$ is a  solution of
(\ref{4.9})
\end{lem}
\begin{proof}
It is easy to check that
\begin{equation*}
\varphi_2^{\pm}-2\mu\varphi_1^{\pm}+1=0.
\end{equation*}
It follows that $\varphi_k^{+}(\mu)\cdot\varphi_s^{-}(\mu)$ satisfies equation (\ref{4.9}).
Indeed, putting this in (\ref{4.9}), we have
\begin{equation*}
\varphi_{\sigma k}^{+}(\mu)\varphi_{\sigma s}^{-}(\mu)\big[(\varphi_2^{+}-2\mu\varphi_1^{+}+1)\varphi_1^{-}+
(\varphi_2^{-}-2\mu\varphi_1^{-}+1)\varphi_1^{+}\big]=0.
\end{equation*}
\end{proof}

\begin{thm} Let $\lambda\in\varrho(-\Delta^c)$. Then the components (\ref{4.5}) of the Green
function are given by
\begin{equation}\label{4.13} G_{k,s,m,n}(\lambda)=\frac{-1}{4R(\mu)}
\left\{\begin{array}{r}\big(\mu+R(\mu)\big)^{|\tau k-m|+|\tau s-n|} \
\mbox{for} \  k=m, \ s>n\\ \mbox{or} \  k>m, \ s=n, \\
                           \big(\mu+R(\mu)\big)^{|\sigma k-m|+|\sigma
                           s-n|} \
\mbox{for} \  k=m, \ s<n\\ \mbox{or} \  k<m, \ s=n, \\
\big(\mu+R(\mu)\big)^{|k-m|+|s-n|} \quad \mbox{for the all others},
                            \end{array}\right.
\end{equation}
where $\mu=1-\frac{\lambda}{4}$.
\end{thm}
\begin{proof}
We must prove that
\begin{equation}\label{4.14}
 <x_{k,s}x_{m,n}, \ (-\Delta^c-\lambda)_xG(x,\tilde{x})>=\delta_{k,s}^{m,n},
 \end{equation}
 where $G(x,\tilde{x})$ is given by (\ref{4.1}).
 Using (\ref{2.19}) and (\ref{4.1}), we can rewrite the left-hand side of (\ref{4.14}) as
 \begin{align}\label{4.15}
 (4-\lambda)G_{k,s,m,n}-G_{\tau k,s,m,n}-G_{\sigma k,s,m,n}-G_{k,\tau s,m,n}-G_{k,\sigma s,m,n}=\notag\\
=4\mu G_{k,s,m,n}-(G_{\tau k,s,m,n}+G_{\sigma k,s,m,n}+G_{k,\tau s,m,n}+G_{k,\sigma s,m,n}).
 \end{align}
 The proof falls naturally into three parts.

 Fix $x_{m,n}$.  First, let  $k=m$ and $s=n$. Substituting (\ref{4.13}) into (\ref{4.15}) and using
 (\ref{4.12}), we obtain
\begin{align*}
<x_{k,s}x_{m,n}, \ (-\Delta^c-\lambda)_xG(x,\tilde{x})>&=
\frac{1}{R(\mu)}(-\mu\varphi_0^{+}(\mu)\varphi_0^{+}(\mu)+\varphi_1^{+}(\mu)\varphi_0^{+}(\mu))=\\
&=\frac{-\mu+\varphi_1^{+}}{R(\mu)}=1.
 \end{align*}
 Note that $\varphi_0^{\pm}(\mu)=(\mu\pm R(\mu))^{0}=1$.

Now we show that (\ref{4.13}) satisfies  equation (\ref{4.8}) for  the cases  $k=m$, \ $s\neq n$ and
 $k\neq m$, \ $s=n$.
Check the case $k=m$, \  $s>n$.
In this case the left-hand side of (\ref{4.14}) is equal to
\begin{align*}
&\frac{-1}{4R(\mu)}\big[4\mu\varphi_1^{+}(\mu)\varphi_{s-n+1}^{+}(\mu)-\varphi_2^{+}(\mu)\varphi_{s-n+1}^{+}(\mu)-
\varphi_0^{+}(\mu)\varphi_{s-n+1}^{+}(\mu)-\\
&-\varphi_1^{+}(\mu)\varphi_{s-n+2}^{+}(\mu)-\varphi_1^{+}(\mu)\varphi_{s-n}^{+}(\mu)\big]=
\frac{\varphi_{s-n+1}^{+}(\mu)}{2R(\mu)}(\varphi_2^{+}-2\mu\varphi_1^{+}+1)=0.
 \end{align*}
 The proof for the case $k=m$, \ $s<n$ (or $k\neq m$, \ $s=n$) is similar.

 Finely, let $k\neq m$, $s\neq n$. We give the proof only for the case $k>m, \ s<n$.
 For the other cases the proof runs similarly.
 In this case we have
 \begin{align*}
G_{k,s,m,n}(\lambda)=\frac{1}{-4R(\mu)}\varphi_{k-m}^{+}(\mu)\varphi_{n-s}^{+}(\mu)=
\frac{1}{-4R(\mu)}\varphi_{k}^{+}(\mu)\varphi_{m}^{-}(\mu)\varphi_{n}^{+}(\mu)\varphi_{s}^{-}(\mu).
\end{align*}
Here we use that $\varphi_{-k}^{+}=\varphi_{k}^{-}$ for any $k\in\mathbb{Z}$.
By Lemma~6,  $\varphi_k^{+}(\mu)\varphi_s^{-}(\mu)$ is a  solution of
(\ref{4.9}) and so is $G_{k,s,m,n}(\lambda)$ at the point $x_{k,s}$.

Thus, since $G(x,\tilde{x},\lambda)$ with components given by (\ref{4.13}) and, by Proposition~5,
 $G(x,\tilde{x},\lambda)\in\mathcal{H}^0$ with respect to $x$, it must be
the Green function of $-\Delta^c: \mathcal{H}^0\rightarrow \mathcal{H}^0$.
\end{proof}

It should be noted that the consideration of $-\Delta^c: \mathcal{H}^2\rightarrow \mathcal{H}^2$
does not differ from that carried out for the 0-forms $\varphi\in\mathcal{H}^0$. In this case we consider
(\ref{4.1}) as a double  form on $\mathfrak{C}^2\times\mathfrak{C}^2$  with basis elements $\Omega^{k,s}\Omega^{p,q}$.
Then we define the Green function as above and its components are given by (\ref{4.13}). The situation with
$-\Delta^c: \mathcal{H}^1\rightarrow \mathcal{H}^1$ is more difficult. In this case, having written equation (\ref{4.7})
pointwise (at the elements $e_{k,s}^1$ and  $e_{k,s}^2$), we obtain a
  pair of
equations (\ref{4.8}) which  corresponds to the components $u_{k,s}, \ v_{k,s}$ of $\omega\in\mathcal{H}^1$.
Roughly speaking, here we must describe the Green function for each component of the 1-form $\omega=(u,v)$.
The reader can verify that the components $G_{k,s,m,n}(\lambda)$ of the Green function have again  the form (\ref{4.13}).

\end{document}